\documentclass[preprint,11pt]{article}
\usepackage{amssymb, amsthm}
\usepackage{graphicx}
\usepackage{algorithm2e}
\usepackage{epstopdf}
\usepackage{url}
\usepackage{color}
\usepackage{refcount}
\usepackage{hyperref}
\usepackage{thmtools,thm-restate}
\usepackage{authblk}

\newtheorem{lemma}{Lemma}
\newtheorem{theorem}{Theorem}
\newtheorem{cor}{Corollary}
\newtheorem{observation}{Observation}

\newcommand{\opt}{\texttt{OPT}}

\makeatletter
 \let\@fnsymbol\@arabic
 \makeatother

\begin{document}
\title{Minimum Forcing Sets for 1D Origami\footnote{E. Demaine supported in part by NSF ODISSEI grant EFRI-1240383 and NSF Expedition grant CCF-1138967. T. Hull supported by NSF ODISSEI grant EFRI-1240441.}}
 
\author{Mirela Damian,\thanks{Dept. of Computer Science, Villanova University, Villanova, PA 19085, USA, mirela.damian@villanova.edu}\ \
Erik Demaine,\thanks{CSAIL,  Massachusetts Institute of Technology, Cambridge, MA 02139, USA, edemaine@mit.edu}\ \ 
Muriel Dulieu,\thanks{mdulieu@gmail.com}\ \ 
Robin Flatland,\thanks{Dept. of Computer Science, Siena College, Loudonville, NY 12211, USA, flatland@siena.edu}\ \ 
Hella Hoffmann,\thanks{David R. Cheriton School of Computer Science, University of Waterloo, Waterloo, ON Canada N2L 3G1, hrhoffmann@uwaterloo.ca}\ \ 
Thomas C. Hull,\thanks{Dept. of Math, Western New England University, Springfield, MA 01119, USA, thull@wne.edu}\ \ 
Jayson Lynch,\thanks{CSAIL,  Massachusetts Institute of Technology, Cambridge, MA 02139, USA, jaysonl@mit.edu}\ \ 
and
Suneeta Ramaswami\thanks{Dept. of Computer Science, Rutgers University, Camden, NJ 08102, USA, rsuneeta@camden.rutgers.edu}}


\date{May 16, 2015}

\maketitle 



\begin{abstract}
This paper addresses the problem of finding minimum forcing sets in origami. 
The origami material folds flat along straight lines called \emph{creases} that can be labeled as mountains or valleys. 
A \emph{forcing set} is a subset of creases that force all the other creases to fold according to their labels. The result is 
a flat folding of the origami material. 
In this paper we develop a linear time algorithm that finds \emph{minimum} forcing sets in one dimensional origami. 
\end{abstract}

\noindent{\bf Keywords:} origami, flat folding, forcing sets


\section{Introduction}

In computational origami, the topic of self-folding origami, where a thin material folds itself in response to some stimulus or mechanism, has been gaining in popularity (see~\cite{Hawkes+10,Ionov13,Leong+07,MR+05}).  The origami material can be programmed to fold along straight lines, called \emph{creases}, by means of rotation. As such, the self-folding process can be economized by only programming a subset of the creases to self-fold, which would then \emph{force} the other, passive creases to fold as originally intended.  We call this subset of creases a {\em forcing set}. Finding a forcing set of smallest size in a given crease pattern would be a useful tool for self-folding.  

This paper addresses the one-dimensional (1D) variant of the forcing set problem, in which the piece of paper reduces to a horizontal line segment in $\mathbb{R}$, and creases reduce to points on the line segment. Our main result is an algorithm that finds a minimum forcing set of any given 1D folding pattern, in time linear in the number of input creases. 
This is the first work to consider finding forcing sets for arbitrary crease patterns.
The forcing set problem was only recently introduced in the literature by Ballinger et al.~\cite{BDE+15} in a paper that considers the problem for one particular class of two-dimensional (2D) foldable patterns called \emph{Miura-ori}. The authors present an algorithm for finding a minimum forcing set of a Miura-ori map in time quadratic in the total number of creases. It is worth noting that the results from~\cite{BDE+15} are specific to the Miura-ori map crease pattern and do not transfer to arbitrary crease patterns (of any dimension), which is what we investigate in this paper for the 1D case. 

Although the forcing set problem was introduced recently and so there is little prior work on it, 
there is much related work on flat-foldable origami, which addresses 
  the problem of determining if a given crease pattern can be folded
  to a flat configuration. 
Arkin et al.~\cite{Arkin04} give a linear time algorithm for determining if
a 1D crease pattern is flat-foldable.  In~\cite{BerHay-SODA-96}, Bern and Hayes show how to determine in linear time whether a general
crease pattern has a mountain-valley assignment for which every node 
is locally flat-foldable, and they prove that deciding whether a crease pattern is globally flat-foldable is NP-hard.  For crease patterns consisting of a regular $m \times n$ grid of squares, the complexity of deciding whether a given mountain-valley assignment can be folded flat (the \emph{map folding problem}) remains open~\cite{DOR07}, although recent progress has been made on $2 \times n$ grids~\cite{Morgan12} and on testing for valid
linear orderings of the faces in a flat-folding~\cite{NW13}. In~\cite{Hull03CountingMV}, Hull gives upper and lower bounds on the number of flat-foldable mountain-valley assignments on a single-node 
crease pattern.

%
In 1D origami, the horizontal line segment representing the piece of paper is oriented, so we can talk about the \emph{top} of the segment (normal pointing upward) and the \emph{bottom} of the segment (normal pointing downward). 

A \emph{crease} is a point on the segment. A \emph{fold} bends the segment $\pm 180^\circ$ about a crease. 
Folds can be of two types, mountain ($M$) and valley ($V$). A mountain (valley) fold about a crease brings together the bottom (top) sides of the left and right segments adjacent to the crease.  We imagine the paper to be sticky on both sides, so that after a fold the paper will stick and cannot be undone.

A \emph{crease pattern} on a segment with endpoints $c_0,  c_n \in \mathbb{R}$  is a strictly increasing sequence of points $C=(c_0, c_1, ..., c_n)$. We refer to $c_1, c_2, \ldots, c_{n-1}$ as \emph{creases}, and to $c_0, c_n$ as \emph{ends}. (Note that $c_0$ and $c_n$ are not creases; we use similar notation for creases and ends for consistency.) For a pair of consecutive creases $c_i$ and $c_{i+1}$, we will refer to $[c_i, c_{i+1}]$ as an \emph{interval}. 

A {\em mountain-valley assignment} (or simply an $MV$ \emph{assignment}) on the crease pattern $C$ is a function $\mu:\{c_1,...,c_{n-1}\} \rightarrow \{$M,V$\}$.  Note that an $MV$ assignment restricts the way in which each crease can be folded. A crease pattern $C$ together with an $MV$ assignment on $C$ forms an \emph{$MV$ pattern}. A \emph{flat folding} of an $MV$ pattern ($C, \mu$) is a piecewise-isometric embedding of the line segment $[c_0, c_n]$, bent along every crease $c_i$, for $0 < i < n$, according to the assignment $\mu(c_i)$, and not bent along any point that is not a crease.  In particular, each interval $[c_i, c_{i+1}]$ of the crease pattern $C$ must be mapped to a congruent copy, and the connectivity between intervals must be preserved. Furthermore, as the segments are folded flat in accordance with the $MV$ assignment, they must not cross one another (although they may overlap). 
%
For a fixed crease pattern $C$, an $MV$ assignment $\mu$ on $C$ is \emph{foldable} if there is a flat folding of the $MV$ pattern ($C, \mu$). Two creases with assignments $M$ and $V$ are said to have \emph{opposite} $MV$ \emph{parity}. 
For a crease pattern $C$, we say that $MV$ assignments $\mu_1$ and $\mu_2$ \emph{agree} with each other on $F \subseteq C$ if 
for each crease $c \in F$, $\mu_1(c) = \mu_2(c)$. 

Given an $MV$ pattern ($C, \mu$), we say that a subset $F$ of $C$ is {\em forcing} if the only foldable $MV$ assignment on $C$ that agrees with $\mu$ on $F$ is $\mu$ itself. This means that, if each crease $c \in F$ is assigned the value $\mu(c)$, then each crease $c' \in C \setminus F$ must be assigned the value $\mu(c')$, in order to produce a foldable $MV$ assignment. 
As an example, consider the $MV$ pattern depicted in~\autoref{fig:crimppair}a, 
and imagine that our crease pattern $C = (c_{i-1}, c_i, c_{i+1}, c_{i+2})$ contains only these four creases. Thus $c_{i-1}$ and $c_{i+2}$ are ends, and the segment folds about $c_i$ and $c_{i+1}$ according to the assignment $\mu(c_i) = M$ and $\mu(c_{i+1}) = V$. 
Assume that the two end intervals $[c_{i-1}, c_i]$ and $[c_{i+1}, c_{i+1}]$ are both strictly longer than the middle interval $[c_i, c_{i+1}]$. Then a forcing set for this example is $F = \{c_{i+1}\}$, because the only foldable assignment that agrees with $\mu(c_{i+1}) = V$ is the one that assigns $c_{i}$ to be an $M$ crease; otherwise, if $c_{i}$ was also a $V$ crease, then the bigger intervals $[c_{i-1},c_i]$ and $[c_{i+1},c_{i+2}]$ would both cover the smaller interval $[c_i,c_{i+1}]$ on the same side, causing a self-intersection (as depicted in~\autoref{fig:crimppair}b). Similar arguments show that $F = \{c_{i}\}$ is also a forcing set, because it forces $c_{i+1}$ to be a $V$ crease, otherwise the assignment would not be foldable. 

A forcing set $F$ is called {\em minimum} if there is no other forcing set with size less than $|F|$.  
In this paper we present a linear time algorithm that takes as input a flat-foldable $MV$ pattern ($C, \mu$) and produces as output a minimum forcing set $F$ of $C$. 

\section{Preliminaries}
This section introduces some terminology and preliminary results that will be used in constructing a minimum forcing set in~\autoref{sec:algorithm} and in proving the correctness of the forcing set algorithm in~\autoref{sec:correct}. 

Throughout the paper we work with a flat-foldable $MV$ pattern $(C, \mu)$, where $C=(c_0, c_1, ..., c_n)$ is a one-dimensional crease pattern and $\mu$ is a foldable $MV$ assignment. So whenever we talk about the $MV$ pattern $(C, \mu)$, we assume that it is \emph{flat-foldable}, unless otherwise specified.  

\subsection{Crimpable Sequences}
\label{sec:crimps}
Given a $MV$ pattern $(C, \mu)$,
our  algorithm (described in~\autoref{sec:algorithm}) repeatedly identifies \emph{crimpable} sequences of creases that can be used to force $MV$ folds. A \emph{crimpable} sequence is composed of consecutive equidistant creases, say distance $d$ apart, with the property that the two intervals adjacent to the left and right end of the sequence are strictly longer than $d$. Formally, for integers $i , k > 0$, a sequence of consecutive creases $(c_i, c_{i+1}, \ldots, c_{i+k})$ is crimpable if 
$|c_{i+1} - c_i| = |c_{i+2}-c_{i+1}| = \ldots = |c_{i+k} - c_{i+k-1}|$ and 
$|c_i-c_{i-1}| > |c_{i+1}-c_i| < |c_{i+k+1}-c_{i+k}|$.

\begin{figure}[hptb]
\centerline{\includegraphics[width=.75\textwidth]{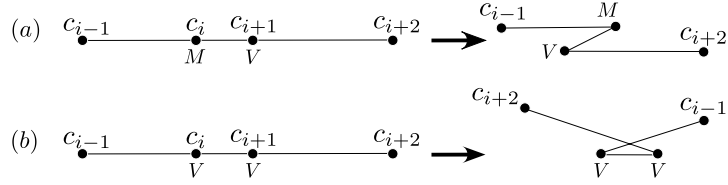}}
\caption{A monocrimp operation on a crimpable sequence $(c_i, c_{i+1})$.  (a) $c_i$, $c_{i+1}$ have opposite $MV$ parity (b) $c_i$, $c_{i+1}$ have the same $MV$ parity; self intersection is unavoidable.}\label{fig:crimppair}
\end{figure}
%

A \emph{monocrimp} operation folds the paper about a single pair of consecutive creases of opposite $MV$ parity in a  crimpable sequence, according to the $MV$ assignment $\mu$. 
~\autoref{fig:crimppair}a depicts a monocrimp operation performed on a pair of creases $(c_i, c_{i+1})$.  
The three intervals involved in the monocrimp operation will fuse together into a new interval with endpoints $c_{i-1}$ and $c_{i+2}$. The 
result is a new crease pattern, reduced from the one before the crimp.  
Note that the notions of top and bottom need to be redefined for the fused segment, so that the top faces upwards and the bottom faces downwards. 
After the monocrimp, the creases $c_i$ and $c_{i+1}$ will disappear and the three intervals from $c_{i-1}$ to $c_{i+2}$ will be replaced with an interval of length 
\begin{equation}
c_i - c_{i-1} - (c_{i+1}-c_i) + c_{i+2}-c_{i+1} = c_{i+2} - 2c_{i+1} +2c_i - c_{i-1}.
\label{eq:crimplen}
\end{equation}

\noindent
One important property of the monocrimp operation, which will be used in our  algorithm, is captured by the following lemma. 
\begin{lemma}
After a monocrimp is performed on a pair of creases $(c_i, c_{i+1})$
in a crimpable sequence, the interval to the right of $c_{i-1}$
(which is also the interval to the left of $c_{i+2}$) can get longer or remain the same size. 
\label{lem:crimpgrow}
\end{lemma}
\begin{proof}
By the definition of a monocrimp, $c_i$ and $c_{i+1}$ belong to a crimpable sequence, therefore
\begin{equation}
c_{i}-c_{i-1} \ge c_{i+1}-c_i \le c_{i+2}-c_{i+1}
\label{eq:monocrimplen}
\end{equation}
The left (right) inequality from~(\ref{eq:monocrimplen}) is strict if $c_i$ ($c_{i+1}$) is the first (last) crease in the crimpable sequence, otherwise the two terms on either side are equal (by the definition of a crimpable sequence). After the monocrimp both $c_i$ and $c_{i+1}$ disappear, and the intervals from $c_{i-1}$ to $c_{i+2}$ are replaced by a single interval of length $c_{i+2} - 2c_{i+1} +2c_i - c_{i-1}$ (see~(\ref{eq:crimplen})). 
Thus the interval to the right of $c_{i-1}$ changes size from $c_{i}-c_{i-1}$ to  
\begin{eqnarray*}
c_{i+2} - 2c_{i+1} +2c_i - c_{i-1} & = & c_i-c_{i-1} + (c_{i+2}-c_{i+1}) - (c_{i+1}-c_i)  \\
						& \ge& c_i-c_{i-1}.   \mbox{~~~~(by~(\ref{eq:monocrimplen}))}
\end{eqnarray*}
This latter inequality follows from the definition of a crimpable pair, by which $c_{i+2}-c_{i+1} \ge c_{i+1}-c_i$. 
An analogous argument holds 
for the interval to the left of $c_{i+2}$. 
%
\end{proof}

\noindent
The following corollary follows immediately from~\autoref{lem:crimpgrow}.
\begin{cor}
Let $\alpha = (c_i, c_{i+1}, \dots, c_{i+k})$ be a crimpable sequence in $C$. If $k > 2$, the result of a monocrimp operation performed on $\alpha$ is a crimpable sequence with two fewer creases. 
\label{cor:intervalgrow}
\end{cor}

Note that a monocrimp operation performed on a crimpable sequence $\alpha = (c_i, c_{i+1})$ requires $c_i$ and $c_{i+1}$ to have 
opposite $MV$ parity, for if they were both $M$ or both $V$, then the bigger intervals $[c_{i-1},c_i]$ and $[c_{i+1},c_{i+2}]$ would both cover the smaller interval $[c_i,c_{i+1}]$ on the same side, causing a self-intersection.  This leads to the key observation that, if $\mu(c_i) = M$, then it must be that $\mu(c_{i+1}) = V$, otherwise $\mu$ is not foldable. Thus we can assign one of the creases in the pair to match the assignment in $\mu$,  and the other crease will be forced to have opposite $MV$ parity. 

We will use a generalization of this property in constructing a minimum forcing set. 
\autoref{thm:crimpforce} below originally appeared in an equivalent form as Theorem~4 in~\cite{Hull03CountingMV}, with angles corresponding to our intervals. It quantifies the difference in the number of $M$ and $V$ assignments for creases in crimpable sequences of even and odd length.
\begin{figure}[hp]
    \centering
   \includegraphics[width=.85\textwidth]{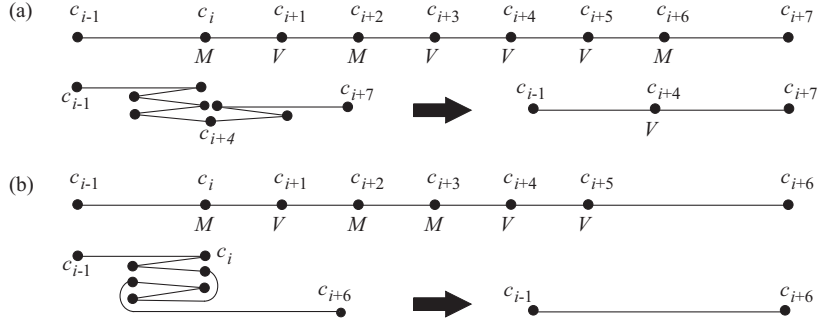}
    \caption{ Crimp operation on an (a)  odd number and (b) even number of creases. Below each crease pattern is the result of the crimp with folds shown (left) and with the folds fused together (right).}
    \label{fig:crimpseq}
\end{figure}

\begin{theorem}[Theorem~4 from~\cite{Hull03CountingMV}]
Let $\alpha$ be a crimpable sequence in a foldable $MV$ pattern. The difference in the number of $M$ and $V$ assignments for the creases in $\alpha$ is zero (one) if $\alpha$ has an even (odd) number of creases. 
\label{thm:crimpforce}
\end{theorem}

\noindent
This theorem allows us to define a \emph{crimp} operation on a crimpable sequence $\alpha = (c_i, c_{i+1}, \dots, c_{i+k})$  as an \emph{exhaustive} sequence of 
monocrimps performed in $\alpha$; here by \emph{exhaustive} we mean that no monocrimp operation can be further performed on the sequence. 
Next we show that a crimp operation can always be performed on a crimpable sequence with two or more creases. 

%
\begin{lemma}
Let $\alpha$ be a crimpable sequence in a foldable crease pattern $C$. 
A crimp operation on $\alpha$ reduces $\alpha$ to a null sequence if
the length of $\alpha$ is even, and to a single crease $c \in \alpha$
if it is odd. In the latter case, the  
position of $c$  and the interval distances within the resulted crease pattern are independent of the $MV$ assignment. 
\label{lem:crimp}
\end{lemma}
\begin{proof}
Let $\alpha = (c_i, c_{i+1}, \dots, c_{i+k})$, for a fixed $k > 0$.
By~\autoref{thm:crimpforce}, among creases $c_i, c_{i+1}, \dots, c_{i+k}$,  there is an adjacent pair with opposite $MV$ assignments. A monocrimp operation on this pair of creases eliminates the creases (thus eliminating one $M$ and one $V$ crease), and the result is a crimpable sequence with two fewer creases (by~\autoref{cor:intervalgrow}). So we can repeat this process until there are zero (if $k$ is odd) or one (if $k$ is even) creases remaining in $\alpha$. In the latter case, the position of this crease is always between $c_{i-1}$ and $c_{i+k+1}$, so it is independent of the $MV$  assignment. As shown in the proof of~\autoref{lem:crimpgrow}, the lengths of the intervals in the resulted crease pattern depend only on the position of the creases in $C$ and therefore are also independent of the $MV$ assignment. 
\end{proof}

\noindent
In the case of a crimpable sequence $\alpha$ of odd length, we say that the crease obtained from $\alpha$ through a crimp operation 
{\em survives} the crimp, and all other creases from $\alpha$ {\em disappear}. The observation below follows immediately from the fact that a crimp operation is composed of a series of monocrimps.

\begin{observation}
\label{obs:b}
If a crimpable sequence $\alpha$ has an odd number of creases, then a crimp operation on $\alpha$ yields a surviving crease with the same $MV$ assignment as the majority of the creases in $\alpha$. 
\end{observation}
The crimp operation is depicted in~\autoref{fig:crimpseq}a for $k=6$, and~\autoref{fig:crimpseq}b for $k=5$. 

\subsection{End Sequences}
\label{sec:endseq}
Our algorithm will repeatedly perform crimp operations, until no crimpable sequences are left in the crease pattern. We call the resulting
crease pattern an \emph{end sequence} and the creases \emph{end creases}.
%
%
\begin{figure}[htp]
    \centering
   \includegraphics[width=.8\textwidth]{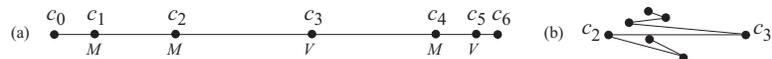}
    \caption{ (a) An end sequence, (b) folded flat using end folds.}
    \label{fig:endsequence}
\end{figure}
%
Let $\lambda$ be the sequence of distances between consecutive creases in an end sequence $E$, in order from left to right.  There can be no local minima in $\lambda$ other than at the ends, because a local minimum would indicate the presence of a crimpable sequence. Therefore, the sequence $\lambda$ starts with zero or more values that are monotonically increasing, and this is followed by zero or more values that are monotonically decreasing. 
%
The following lemma implies that no $MV$ assignments can be forced on
the creases of an end sequence, which also implies that the forcing set of an end sequence must include all the end creases. 


\begin{lemma}
Any $MV$ assignment to the creases of an end sequence is foldable. 
\label{lem:endfold}
\end{lemma}

 \begin{proof}
We show how to fold an end sequence $E$ with an arbitrary $MV$ assignment 
using a series of {\it end folds}. An end fold consists of folding the first or last interval of the sequence under or over, depending 
on the $MV$ assignment of the crease. To ensure that the folded interval doesn't overlap other creases, an end fold is performed 
only when the adjacent interval is of the same length or longer. 
Note that each end fold can be performed regardless of its $MV$ assignment. 

Let $\ell = [c_j,c_{j+1}]$ be the longest interval in $E$; in case of ties pick the leftmost interval.
Because the intervals are monotonically increasing to the left of $\ell$  and  monotonically decreasing to the right of $\ell$,
we can repeatedly fold the leftmost interval under (for a mountain fold) or over (for a valley fold)
until $\ell$ becomes the leftmost interval. Similarly, because the intervals are monotonically increasing to the right 
of $\ell$, we can  do the same on the right end of $E$
until only $\ell$ remains, thus producing a flat folding. See
\autoref{fig:endsequence} for an example where $\ell = [c_2,c_3]$. 
\end{proof}

\noindent
It follows from~\autoref{lem:endfold} that the forcing set of an end sequence must include all the end creases. 

\section{Constructing a Minimum Forcing Set}
\label{sec:algorithm}
We describe a simple algorithm that finds a minimum forcing set $F$ for any given foldable $MV$ pattern $(C, \mu)$.
Although the algorithm is straightforward, the proof that it finds a minimum forcing set is quite intricate. 
The algorithm begins by 
constructing a forest of trees whose nodes correspond to crimpable sequences encountered during the
folding process and whose edges capture dependencies between them (\autoref{sec:forest}). 
Then it traverses the trees in a top-down manner to 
decide which creases corresponding to each node should be added to the 
forcing set  (\autoref{sec:force}).

\subsection{Crimp Forest Construction}
\label{sec:forest}
The {\sc CrimpForest} algorithm detailed in~\autoref{alg:crimpForest} computes a {\it crimp forest} $W$ corresponding to the foldable input $MV$  pattern $(C,\mu)$. $W$ is a collection of rooted trees, where each node corresponds to a crimpable sequence that is encountered during the folding process.
For simplicity, we will use the terms ``node'' and ``crimpable sequence'' interchangeably to refer to a node in a crimp tree.
 For a crease pattern $C$, we can easily identify its set $S$ of crimpable sequences by scanning its interval lengths 
from left to right.  
In $W$, a node $c$ is a child of another node $p$ if 
the crimp operation on $c$ has a surviving crease, and the next crimp sequence involving the surviving crease is $p$.

\begin{algorithm}
\centerline{{\sc CrimpForest}($C, \mu$)}
{\hrule width 0.92\linewidth}\vspace{0.8em}
Initialize $W \leftarrow \emptyset$.\\
\While{$C$ has a crimpable sequence}{
      Let $s$ be the leftmost crimpable sequence in $C$. \\
      Create a node $v$ corresponding to $s$, and add $v$ to $W$. \\
      Make $v$ the parent of each root node in $W$ whose crimp sequence has a surviving crease that is in $s$. \\ 
      Apply the crimp operation to $s$. \\
      Update $C$ to be the resulting crease pattern. 
}
\vspace{2mm}\noindent{\hrule width 0.92\linewidth} \vspace{1mm} 
\caption{\sc CrimpForest algorithm.}
\label{alg:crimpForest}
\end{algorithm}

See~\autoref{fig:forestex} for an example of a crease pattern and corresponding crimp forest. Initially the forest $W$ is empty.
The leftmost crimpable sequence in the crease pattern from~\autoref{fig:forestex} (bottom right) is $\alpha_1 = (c_1, c_2, c_3)$.
 The {\sc CrimpForest} algorithm processes $\alpha_1$ by adding a corresponding node to the forest $W$ and performing the crimp operation on $\alpha_1$, resulting in the smaller crease pattern shown on the right, second from bottom. The crease $c_1$ survives the crimp operation on $\alpha_1$. The leftmost crimpable sequence is now $\alpha_2 =  (c_1, c_4, c_5)$.
Because $\alpha_2$ includes the survivor $c_1$ of $\alpha_1$, the algorithm adds $\alpha_2$ to $W$ as a parent of $\alpha_1$.
After crimping $\alpha_2$, the leftmost crimpable sequence is $\alpha_3 = (c_6, c_7, c_8)$, which is crimped (with $c_8$ surviving) and a corresponding node is added to $W$.
The new crease pattern includes one crimpable sequence $\alpha_4 = (c_5, c_8, c_9)$. Because $c_5$ is the survivor of $\alpha_2$ and $c_8$ is the survivor of $\alpha_3$, the node corresponding to $\alpha_4$ becomes parent of both $\alpha_2$ and $\alpha_3$ in $W$. 
The underlined creases from~\autoref{fig:forestex} will be used in the description of the minimum forcing set algorithm from~\autoref{sec:force}.

%
\begin{figure}[htp]
    \centering
   \includegraphics[width=\textwidth]{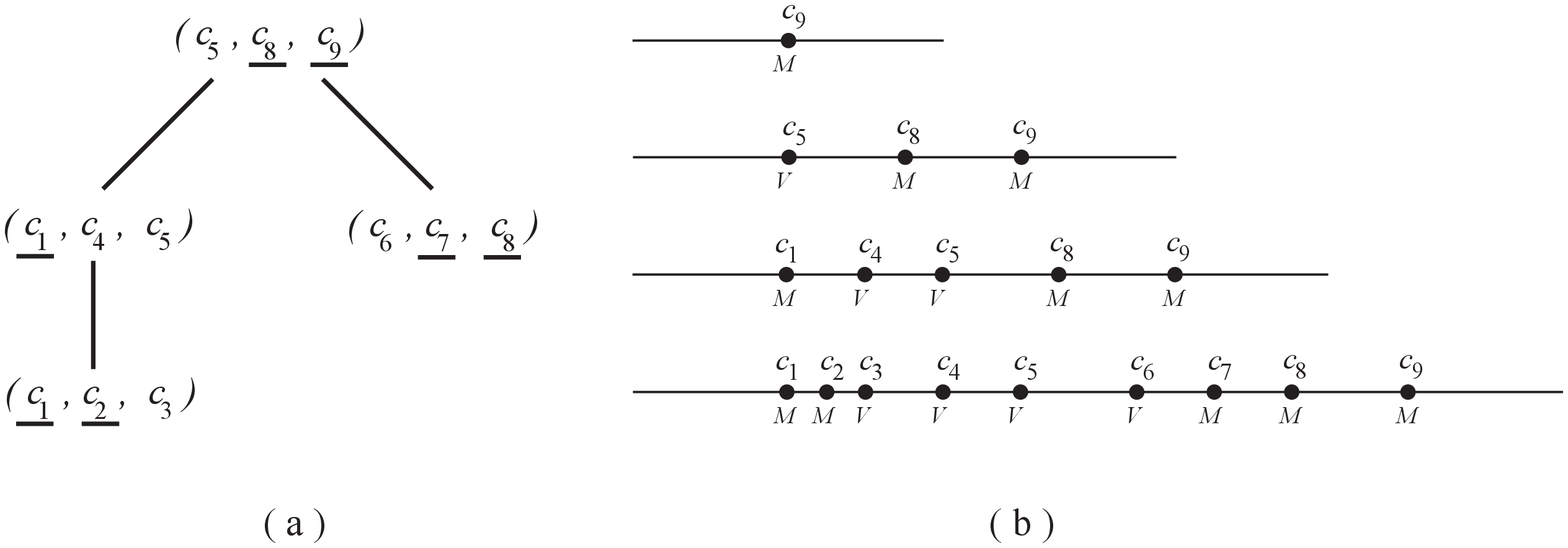}
    \caption{ The (a) crimp forest corresponding to the (b) crease pattern $C$ on the lower right having
interval lengths 4, 1, 1, 2, 2, 3, 2, 2, 3, 4.
Each new crease pattern resulting after processing the leftmost crimpable sequence
is shown above the previous pattern. 
The underlined creases in (a) are referred to in the algorithm from~\autoref{sec:force}.
}
    \label{fig:forestex}
\end{figure}


The observation below follows immediately from the fact that creases at the root nodes are not involved in any further crimp operations. 
\begin{observation}
\label{obs:a}
If a crimpable sequence $\alpha$ has an even number of creases (and so an equal number of $M$ and $V$ assignments), then $\alpha$ is the root of a tree in $W$. If the root $r$ of a tree in $W$ has an odd number of creases (and so unequal number of $M$ and $V$ assignments), then the crimp operation on $r$ yields a surviving crease that is an end crease (in the end sequence). 
\end{observation}

\noindent
We will need the following lemma in our proof of correctness in~\autoref{sec:correct}. 

\begin{lemma}
Given a crease pattern $C$ and two foldable $MV$ assignments $\mu_1$ and $\mu_2$, let $W_1$ and $W_2$ be the crimp forests corresponding to $(C,\mu_1)$ and $(C,\mu_2)$, respectively.  
Then the following properties hold:
\begin{enumerate}
\item [\emph{(1)}] $W_1$ and $W_2$ are structurally identical. 
\item [\emph{(2)}] Corresponding nodes in $W_1$ and $W_2$ have crimpable sequences of the same size and the same interval distances between adjacent creases.
\item [\emph{(3)}] Creases involved for the first time in a crimpable sequence at a node in $W_1$ have the same 
position in the crimpable sequence at the corresponding node in $W_2$. 
\end{enumerate}
\label{lem:forest}  
\end{lemma}
\begin{proof}
By definition, a consecutive sequence $c_i, c_{i+1}, \dots, c_{i+k}$ of creases is crimpable if the creases are equally spaced at some distance $d$ and the intervals adjacent to the right of $c_{i+k}$ and to the left of $c_i$ are strictly longer than $d$.  The definition is based solely on interval lengths, and so the identification of crimpable sequences is independent of the $MV$ assignment. 

The sequence of interval lengths resulting after a crimp operation is also independent of the $MV$ assignment. This is because a crimp operation consists of a series of monocrimps, and after each monocrimp two adjacent creases $c_i$ and $c_{i+1}$ disappear and are replaced with an interval of length $c_{i+2}-2c_{i+1} + 2c_i -c_{i-1}$, as derived in equation~(\ref{eq:crimplen}) in~\autoref{sec:crimps}.  It does not depend on the $MV$ assignment.
We note, however, that for an odd length crimpable sequence, the particular crease that survives may be different. For example,
in crimpable sequence $(c_1,c_2,c_3)$ in~\autoref{fig:forestex}a with assignment $(M,M,V)$, crease $c_1$ survives; but with assignment $(V,M,M)$, crease $c_3$ survives instead.

Consider now two runs of the {\sc CrimpForest} algorithm (\autoref{alg:crimpForest}), one using input $(C,\mu_1)$ and the other using input $(C, \mu_2)$. In the first iteration of the while loop, both runs identify the same set of crimpable sequences in $C$, because identification of crimpable sequences is independent of $MV$ assignment. In this case nodes with identical crimpable sequences of creases are added to $W_1$ and $W_2$.
(We abuse our notation slightly here by letting $W_1$ be the forest in the current iteration of the run using input $(C,\mu_1)$, and similarly for $W_2$). Therefore, properties (1), (2) and (3) hold for $W_1$ and $W_2$ after the first iteration of the two runs.  
In addition, after performing the crimp operation on each identified crimpable sequence, the updated crease pattern $C$ in both runs consists of the same sequence of interval lengths, because the interval sequence resulting after a crimp operation is independent of the $MV$ assignment. As noted earlier, the particular crease surviving a crimp operation in one run may not be the same as the surviving crease in the other run. By~\autoref{lem:crimp}, for each survivor in one run, there is a corresponding survivor in the other and they have the same 
position in the resulting crease sequence. This will be important later when adding edges to the forests. 

Assume inductively that after the $i$th iteration,  properties (1), (2) and (3) hold for $W_1$ and $W_2$. 
In addition, assume inductively that the resulting crease pattern in each run consists of the same sequence of interval lengths, and includes all creases not yet involved in any crimpable sequences in the same 
position (in addition to those creases surviving the crimps).
 
Now consider the $(i+1)$th iteration. Because the crease patterns in both runs after iteration $i$ have the same sequence of interval lengths, the set of crimpable sequences identified in iteration $i+1$ will also be the same. 
Thus each node created in iteration $i+1$ of the run under $\mu_1$ has a corresponding node created in the run under $\mu_2$ with a crimpable sequence of the same size and same interval lengths. Thus property (2) holds. 
By the inductive hypothesis, creases involved for the first time in the crimpable sequences identified in iteration $i+1$ have the same 
position in the crease patterns resulted from iteration $i$. In particular, they have the same 
position in the crimpable sequences associated with corresponding nodes in $W_1$ and $W_2$, thus establishing property (3).

If a node $n_1$ created in iteration $i+1$ of the run under $\mu_1$ contains a crease $c_j$  that is a survivor of the crimpable sequence at a root node $r_1$ of $W_1$, then by the inductive hypothesis, there is a corresponding root node $r_2$ in $W_2$ whose crimpable sequence has a corresponding survivor $c_\ell$. 
In the $(i+1)$th iteration of the second run, a node $n_2$ (corresponding to $n_1$)
is created that contains $c_\ell$ in its crease sequence.
Therefore,  edge $(n_1,r_1)$ will be added to $W_1$ and edge
$(n_2,r_2)$ will be added to $W_2$, so the structure of the two forests remains the same. Thus property (1) holds. 
Because the identified crimpable sequences are in the same locations, and have the same size and same interval
lengths in both runs, after applying the crimp operations in iteration $i+1$, the resulting
crease sequences have the same interval distances in both runs. By~\autoref{lem:crimp}, survivors of 
these crimp operations will have the same 
position in the two crease sequences. This along with the inductive hypothesis implies that creases 
not yet involved in a crimp operation, which intersperse with survivors of previous crimp operations, 
have the same 
position in the crease patterns resulted from the two runs. 
This establishes the inductive hypothesis for the $(i+1)$th iteration. 
\end{proof}

\begin{lemma}
The time complexity of the {\sc CrimpForest} algorithm (\autoref{alg:crimpForest}) is $O(n)$, where $n$ is the number of input creases. 
\label{lem:crimpforest-time}
\end{lemma}
\begin{proof}

We show that the {\sc CrimpForest} algorithm has a simple linear-time implementation that scans the crease pattern $C$ from left to right, repeatedly identifying the leftmost crimpable sequence, $\alpha$, and then processing and crimping it as described in~\autoref{alg:crimpForest}. The only complication is that crimping $\alpha$ can result in a new crimpable sequence forming on its left, but we show this can be detected and handled efficiently. 

We describe the details in three steps which we repeat until we
  reach the end of $C$. It is helpful here to recall that a crimpable
  sequence is a local minimum in the interval lengths of $C$. 
  \begin{itemize}
    \item In step 1, we scan the intervals in $C$ from left to right
      searching for the first local maximum. Specifically, we pass
      over zero or more intervals of monotonically increasing length
      until reaching the first interval, $I_0$, having a smaller
      interval immediately to its right.
    \item Then in step 2, we resume scanning to the right for the
      first local minimum, passing over one or more intervals of
      monotonically decreasing length until reaching the first
      interval having a larger interval immediately to its right. This
      interval marks the end of the leftmost crimpable sequence,
      $\alpha$, in $C$.  
    \item In step 3, in time linear in the number of creases in
      $\alpha$, we can identify $\alpha$'s creases and the intervals
      $I_\ell$ and $I_r$ located to its left and right, and then process
      $\alpha$ as described in the loop iteration
      of~\autoref{alg:crimpForest}.
  \end{itemize}

If $\alpha$ has an odd number of creases, then $I_\ell$'s length is
  not affected by the crimp operation on $\alpha$, so we go back to
  step 2 and resume scanning for the first local minimum starting from
  $I_\ell$. If, on the other hand, $\alpha$ has an even number of
  creases, then $I_\ell$ and $I_r$ merge into a new, longer interval $I$
  after the crimp operation, which may have produced a new crimpable
  sequence (i.e., local minimum) immediately to the left of $I$. 
  There are three cases to consider: (a) If 
  $I_0\subset I$ (i.e.,  $I_0 = I_\ell$ merged with $I_r$ into $I$),  
 or if $I_0$ is
  immediately to the left of $I$ and $I$ is longer than $I_0$, then no
  new local minimum was produced, but the first local maximum may have
  changed; in this case, we go back to step 1 and resume scanning from
  $I$ for the first local maximum (and reset $I_0$ if the maximum
  changed); (b) If $I$ is shorter than the interval $J$ immediately to
  its left, then again no new local minimum was produced and we go
  back to step 2, looking for the first local minimum starting from
  $I$; (c) If however $I$ is longer than $J$, then $J$ marks the end
  of a newly formed local minimum (because interval lengths from $I_0$
  to $J$ have decreasing lengths and $I$ is longer than $J$); in this
  case, we let $\alpha$ be the newly formed crimpable sequence, and we
  go back to step 3 and process it similarly. 

Using this implementation, each interval is involved in only a constant number of operations that include comparisons and crimps. This shows that the crimp forest construction can be done in time linear in the number of creases in $C$.
\end{proof}

\subsection{Forcing Set Algorithm}
\label{sec:force}
Given a foldable $MV$ pattern $(C,\mu)$, the {\sc ForcingSet} algorithm described in~\autoref{alg:forceSet} computes a minimum forcing set $F$ of $C$. The algorithm starts by constructing the crimp forest discussed in the previous section (\autoref{sec:forest}), then processes each tree in the forest in a top down manner. 
%
As an example, the forcing set for the crimp tree depicted in~\autoref{fig:forestex} consists of all underlined creases in the tree. Initially the forcing set $F$ contains the only end crease $c_9$. We begin a preorder traversal of the tree rooted at node $\alpha_4 = (c_5, c_8, c_9)$.  Because the surviving crease $c_9$ of $\alpha_4$ is already in $F$, we add to $F$ the crease $c_8$, because it has the majority $MV$ assignment (same as $c_9$). Moving down to  $\alpha_2 = (c_1, c_4, c_5)$, the surviving crease $c_5$ of $\alpha_2$ is not in $F$, so in this case we add $c_1$ to $F$ because it has the minority $MV$ assignment among all creases in $\alpha_2$. The other nodes are handled similarly.  This procedure takes $O(n)$ time for a crease pattern with $n$ creases. 

\begin{algorithm}
\centerline{{\sc ForcingSet}($C, \mu$)}
{\hrule width 0.94\linewidth}\vspace{0.8em}
Initialize $W$ to the output generated by {\sc CrimpForest}($C, \mu$). \\
Initialize $F$ to the set of end creases that remain after running {\sc CrimpForest}$(C,\mu)$ \\ 
\For {\textnormal {\bf {each}} tree $T \in W$ } {
   \For {\textnormal {\bf {each}} node $v$ in a preorder traversal of $T$  } {
       \uIf {$v$'s crimp sequence has even length} {
           Add to $F$ all creases from $v$'s crimp sequence having $M$ assignment. \\
       }
       \uElseIf{the surviving crease from $v$'s crimp sequence is already in $F$} {
                 Add to $F$ all creases from $v$'s crimp sequence having the majority $MV$ assignment.\\
        }
        \Else{
                 Add to $F$ all creases from $v$'s crimp sequence having the minority $MV$ assignment.\\ 
         }   
   }
}
\vspace{2mm}\noindent{\hrule width 0.94\linewidth} \vspace{1mm} 
\caption{Forcing set algorithm.}
\label{alg:forceSet}
\end{algorithm}

\noindent
\autoref{lem:force} is a key ingredient in our proof of correctness (presented in~\autoref{sec:correct}).
\begin{lemma}
Let $(C, \mu_1)$ be a foldable $MV$ pattern, and let $F$ be the forcing set generated by~\autoref{alg:forceSet} with input $(C,\mu_1)$. Let $(C, \mu_2)$ be a foldable pattern such that $\mu_2$ agrees with $\mu_1$ on the forcing set $F$.  Let 
$T_1$ and $T_2$ be two structurally equivalent trees generated by the forcing set algorithm with inputs $(C,\mu_1)$ and $(C,\mu_2)$, respectively.  
If a crease $c$ in a crimpable sequence $\alpha_1 \in T_1$ is in $F$, then a crease (not necessarily $c$) with the same $MV$ assignment occurs in the corresponding crimpable sequence  $\alpha_2 \in T_2$, in the same 
position as in $\alpha_1$. 
\label{lem:force}
\end{lemma}
\begin{proof}
Our proof is by induction on the height $h$ of the trees $T_1$ and $T_2$. The base case corresponds to $h=1$ (each tree is a leaf node). In this case $\alpha_1$ is the first crimpable sequence involving $c$, and the lemma follows immediately  
from property (3) established by~\autoref{lem:forest}. 

The inductive hypothesis states that the lemma holds for any trees $T_1$ and $T_2$ of height $i$ or less. If $\alpha_1$ is at level $i$ or lower in $T_1$, then the lemma holds by the inductive hypothesis. So consider now the case when $\alpha_1$ is at level $i+1$ and contains crease $c \in F$. If $c$ has not yet been involved in a crimp operation, then the lemma follows immediately from property (3) established by~\autoref{lem:forest}. 
Otherwise, $c$ is a crease surviving a crimp operation performed on some child node $\beta_1$ of $\alpha_1$. Because the forcing set algorithm processes $T_1$ top-down, the crease $c$ is already in $F$ at the time node $\beta_1$ is processed. As a result, all creases from $\beta_1$ with the majority assignment are added to $F$. By the inductive hypothesis, creases with the same $MV$ assignment occur in the corresponding node $\beta_2 \in T_2$, in the same 
positions as in $\beta_1$. Because $(C, \mu_1)$ and $(C, \mu_2)$ are foldable, by~\autoref{thm:crimpforce} all other creases in $\beta_1$ and $\beta_2$ have the opposite assignment. 
This implies that the crease surviving the crimp operation on $\beta_2$ has the same $MV$ assignment as the crease surviving the crimp operation on $\beta_1$, which is $c$. This shows that a crease with the same assignment as $c$ occurs in $\alpha_2$ as well. By~\autoref{lem:crimp}, this crease has the same 
position in $\alpha_2$ as $c$ in $\alpha_1$. 
\end{proof}

\section{Proof of Correctness}
\label{sec:correct}
In this section we prove that the {\sc ForcingSet} algorithm outlined in~\autoref{sec:force} produces a minimum forcing set $F$ of the input $MV$ pattern $(C, \mu)$.  
Call two crimpable sequences $\alpha_1$ and $\alpha_2$  \emph{similar} if they have the same size, the same $MV$ assignment read from left to right, and same interval lengths. 
Call a sequence of crimps $O$ \emph{exhaustive} if the crease pattern produced by $O$  contains no crimpable sequences (so no more crimps could be performed on this crease pattern). Our algorithm will perform an exhaustive sequence of crimps on the crease pattern $C$.  An important issue here is that, if several crimps could be performed on the crease pattern at one point, then we must choose a particular order for these crimps, including any new crimps that might occur as a result of the previous crimps. A natural question here is whether this ordering is important. In other words, could different orderings of crimps result in different (say, larger or smaller) forcing sets?~\autoref{lem:order} below proves that this will not be an issue. 

\begin{restatable}{lemma}{orderlemma}
\label{lem:order}
Let $O_1, O_2$ be two non-identical, exhaustive sequences of crimps performed on a one-dimensional crease pattern $C$. Then every crimp that is in $O_1$ is also in $O_2$.
\end{restatable}

\autoref{lem:order} is independent of the algorithm itself, so to maintain the flow of presentation we defer a proof of this lemma to the Appendix. 
Our proof of correctness uses two more fundamental lemmas. 

\begin{lemma}[\cite{DOR07}]
Any folding of a foldable 1D crease pattern can be performed by a sequence of monocrimps and end folds.
\label{lem:foldbycrimps}
\end{lemma}

\begin{lemma}
Let $(C, \mu)$ be a foldable $MV$ pattern, and let $(a_i, a_{i+1})$ be a pair of adjacent creases of opposite $MV$ parity in a crimpable sequence of $C$. Let $\mu'$ be identical to $\mu$, with the exception that $\mu'(a_i) = \mu(a_{i+1})$ and $\mu'(a_{i+1}) = \mu(a_i)$. Then $\mu'$ is a foldable $MV$ assignment.     
\label{lem:switch}
\end{lemma}
\begin{proof}
A flat folding of $(C, \mu')$ is identical to a flat folding of $(C, \mu)$, with the only exception that the overlap order of the two layers corresponding to the creases $a_i$ and $a_{i+1}$ gets switched. 
\end{proof}

\noindent
Our main result is stated by the following theorem. 

\begin{theorem}
Given a flat-foldable $MV$ pattern ($C, \mu$), the {\sc ForcingSet} algorithm (\autoref{alg:forceSet}) produces a minimum forcing set $F$ of ($C, \mu$) in time $O(n)$, where $n$ is the number of creases in $C$. 
\label{thm:main}
\end{theorem}
\begin{proof}
By~\autoref{lem:crimpforest-time}, the {\sc CrimpForest} algorithm constructs the crimp forest $W$ in linear time. The {\sc ForcingSet} algorithm then processes each crimpable sequence $s$ corresponding to a node in $W$ in time linear in the length of $s$. Since the total length of all crimpable sequences in the forest is $O(n)$, we obtain linear time complexity for the {\sc ForcingSet} algorithm.  

Our proof of correctness has two parts:  first we prove that $F$ is indeed a forcing set for $(C, \mu)$, then we prove that $F$ has a minimum size. 

\paragraph{$F$ is forcing} Assume to the contrary that there exists a different foldable $MV$ assignment $\mu_2$ for $C$ that agrees 
with $\mu$ on $F$. For symmetry, let $\mu_1 = \mu$ and run the {\sc ForcingSet} algorithm with input $(C,\mu_1)$ to create forest $W_1$, then run it again with input $(C,\mu_2)$ to create forest $W_2$. By property (1) of~\autoref{lem:forest}, $W_1$ and $W_2$ are structurally equivalent.

Consider first the case when there is a pair of corresponding trees $T_1 \in W_1$ and $T_2 \in W_2$ that have
at least one pair of corresponding nodes with dissimilar crimpable sequences.
By~\autoref{lem:forest}, their crimpable sequences involve the same number of creases, have the same interval
lengths, and first-time creases 
have the same 
locations in both. Therefore the sequences differ only in their $MV$ assignments.
Of all such pairs, let $v_1 \in T_1$ and $v_2 \in T_2$ be a pair of maximal depth in their trees.

Let $\ell$ be the length of the crimpable sequences corresponding to $v_1$ and $v_2$. If $\ell$ is even, then $v_1,v_2$ are root nodes in $T_1,T_2$
(by~\autoref{obs:a}).
In this case, the {\sc ForcingSet} algorithm places in $F$ the $\ell/2$ creases of $v_1$ with assignment $M$. Then by~\autoref{lem:force}, 
there are $\ell/2$ creases in $v_2$ in the same positions that are also assigned $M$. 
By~\autoref{thm:crimpforce}, the remaining $\ell/2$ creases in both $v_1,v_2$ must have $V$ assignments. 
But then $v_1,v_2$ have similar crimpable sequences, a contradiction.

Now consider the case when $\ell$ is odd. If the creases of $v_1$ with the majority $MV$ assignment are in $F$, then
by~\autoref{lem:force} $v_2$ has creases of the same $MV$ assignment located in the same 
positions. All other creases in $v_1,v_2$ must have the opposite assignment by~\autoref{thm:crimpforce}.
But then $v_1,v_2$ are similar, a contradiction. 

\begin{figure}[htp]
    \centering
   \includegraphics[width=\textwidth]{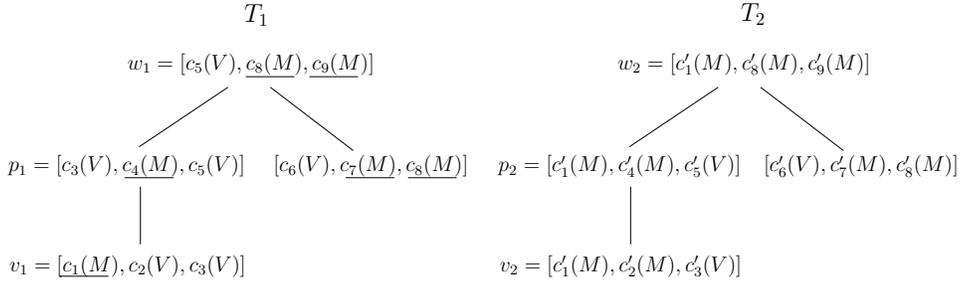}
    \caption{ Proof of correctness for the {\sc ForcingSet} algorithm (\autoref{alg:forceSet}) from~\autoref{sec:force}. Underlined are the creases included in the forcing set $F$.}
    \label{fig:proofex}
\end{figure}

If instead the creases of $v_1$ with the minority $MV$ assignment are in $F$, then assume without loss of generality that the minority
assignment is $M$. This case is depicted in~\autoref{fig:proofex}. The survivor of $v_1$ is not in $F$ (because otherwise the creases of $v_1$ with the majority $MV$ assignment would be in $F$, as indicated by the second branch of the conditional in the {\sc ForcingSet} algorithm), and it has a $V$ assignment (because the survivor always has the majority assignment, cf.~\autoref{obs:b}). 

Next we show that on the path from $v_1$ to the root $r_1$ of $T_1$, we must encounter a node with majority assignment
$M$, or an equal number of $M$ and $V$ assignments. 
Suppose that this is not the case. Then 
each node $u$ on the path from $v_1$ to $r_1$ has majority $V$ assignments, and a crimp operation on $u$  produces a surviving crease of type $V$.  In particular, the crease surviving the crimp operation on $r_1$ is an end crease (cf.~\autoref{obs:a}), and therefore is in $F$ (cf. step 2 of the forcing set algorithm).  But in this case the {\sc ForcingSet} algorithm would place in $F$ the majority $V$ creases from $r_1$, and would do the same at each node on the path from $r_1$ down to and including node $v_1$ (because the crimp survivor at each such node would be in $F$ at the time the node would be processed). But this contradicts our assumption that the $M$ creases from $v_1$ are in $F$.
So let $w_1$ be the first node encountered on the path from $v_1$ to $r_1$ having majority assignment
$M$ or an equal number of $M$ and $V$ assignments. (See the node labeled $w_1$ in~\autoref{fig:proofex}.) 

By~\autoref{lem:force}, if the minority $M$ creases of $v_1$ are in $F$, then $v_2$
has $(\ell-1)/2$ creases of type $M$ located in the same 
positions. 
Therefore, any differences between $v_1,v_2$ must be among the other $(\ell+1)/2$ creases. 
Because $v_1, v_2$ have maximal depth, the crimpable sequences at
corresponding nodes in their children's subtrees are similar, and thus their
children contribute to $v_1$ and $v_2$ survivors with the same $MV$ assignment. 
The differences in $v_1,v_2$'s crimpable sequences must therefore be in first-time creases.
(Note that if $v_1,v_2$ are leaves, then all their creases are first-time creases.)

We now show that all such first-time creases in $v_2$ must have a $V$ assignment.
Suppose this is not the case (as illustrated in~\autoref{fig:proofex} where
$\mu(c_2') \neq V$). Then $M$ is the majority assignment at $v_2$, and
the survivor of $v_2$ is of type $M$ -- in contrast to $v_1$, whose
survivor is of type $V$. (In~\autoref{fig:proofex},  the survivor is 
$c'_1$ of type $M$). Consider the parents $p_1,p_2$
of $v_1,v_2$, respectively. If $p_1 \neq w_1$, then $p_1$ has majority assignment $V$. In addition, its minority
creases of type $M$ are in $F$ --- for if its majority $V$ creases were in $F$, then
the majority $V$ creases at $v_1$
would also be in $F$,
a contradiction. By~\autoref{lem:force},
each $M$ crease in $p_1$ has a corresponding $M$ crease in
$p_2$. In addition, $p_2$ has an additional crease of type $M$ propagated by $v_2$,
making $M$ the majority assignment at $p_2$. Thus $p_2$ produces a survivor of type $M$ -- in contrast to
$p_1$, which has a survivor of type $V$. We can apply this
argument at each node on the way back up the tree to the child of node $w_2$ (corresponding
to $w_1$). We recycle our notation here by refering to $w_2$'s child as $p_2$ and the corresponding node
in $T_1$ as $p_1$. 

If $w_1$ has equal $M$ and $V$ assignments, then it is the root node (see~\autoref{obs:a}),
its $M$ creases are in the forcing set, and corresponding $M$ creases are also in $w_2$.
In addition, $w_2$ has an additional crease of type $M$
propagated by $p_2$. But then the difference in $M$ and $V$
creases at $w_2$ is greater than one, contradicting~\autoref{thm:crimpforce}. 
If, on the other hand, $w_1$ has majority assignment $M$, then its $M$ creases must be in $F$ 
(otherwise the majority $V$ creases at each node from $p_1$ down to and including $v_1$ would also be in $F$, contradicting our assumption). Therefore, $w_2$ has corresponding $M$ creases. In
addition, $w_2$ also contains an additional crease of type $M$ that it got from $p_2$. But then the difference in $M$ and $V$
creases at $w_2$ is greater than one, which again contradicts~\autoref{thm:crimpforce}.

We have established that corresponding first-time creases, as well as survivors from children of $v_1$ and $v_2$, have the same $MV$ assignment in $v_1$ and $v_2$. This implies that $v_1$ and $v_2$ are similar, contradicting our choice of dissimilar sequences.

Finally, consider the case when all corresponding nodes in $W_1$ and $W_2$ have similar crimpable sequences. Then there must be
creases in $(C,\mu_2)$ not involved in any crimp operation, 
whose assignments differ from the corresponding 
creases in $(C, \mu_1)$. Such creases must be part of the end sequence
of $(C,\mu_2)$. However~\autoref{lem:order} implies that end sequence creases are independent on 
the $MV$ assignment. 
Because all end sequence creases are in $F$, it is not possible that $\mu_1$ and $\mu_2$ differ on end sequence creases, so this case is settled. It follows that $\mu_1$ and $\mu_2$ are identical and therefore $F$ is a forcing set. 

\paragraph{$F$ is minimum} Let \opt\ be a forcing set of $(C, \mu)$ of minimum size. Because \opt\ is a forcing set, all other creases in $C \setminus \opt$ must have an $MV$ assignment that agrees with $\mu$. 
By~\autoref{lem:foldbycrimps}, any folding of $(C, \mu)$ can be achieved by a series of monocrimps and end folds. By~\autoref{lem:order}, any ordering of the crimp operations produces similar sets of crimpable sequences, and similar end sequences. 
This implies that each monocrimp involved in the {\sc CrimpForest} algorithm appears in any set of monocrimps used to produce a flat folding of $(C, \mu)$. 
 
Let $m$ be the number of monocrimps involved in the crimps performed by the {\sc ForcingSet} algorithm, and let $e$ be the size of the end sequence. We show that $|\opt| = m+e$. First note that \opt\ must include at least one crease from each monocrimp; otherwise, we can switch the $MV$ assignment for the creases involved in the monocrimp and obtain a foldable $MV$ assignment that is different from $\mu$ (\autoref{lem:switch}), contradicting the fact that \opt\ is a forcing set. Also by~\autoref{lem:endfold}, \opt\ must include all the creases from the end sequence. The intersection between the set of creases involved in the monocrimp operations and the set of creases in the end sequence is empty, because each monocrimp removes the pair of creases involved. It follows that 
$|\opt| = m+e$.

We now show that $|F| = m+e$ as well. The end sequence added to $F$ in the second step of the algorithm contributes the second term to this equality. We claim that, corresponding to each crimpable sequence $\alpha$ of size $\ell$, the {\sc ForcingSet} algorithm adds to $F$ precisely $\lfloor\ell/2\rfloor$ creases. This is clear for the cases when $\alpha$ is of even length, and  when the minority of the creases  from $\alpha$ are added to $F$. Now note that the {\sc ForcingSet} algorithm adds to $F$ the majority creases from $\alpha$  only if the survivor of $\alpha$ is already in $F$. This means that the additional contribution of $\alpha$ to $F$ is again $\lfloor\ell/2\rfloor$ (because one of the creases from the majority is already in $F$). The quantity $\lfloor\ell/2\rfloor$ is precisely the number of monocrimps involved in crimping $\alpha$. Summing up over all crimps performed by the algorithm, we get the total number of creases contributed to $|F|$ by the crimpable sequences to be $m$ as well. This proves that $F$ is a minimum forcing set. 
\end{proof}

The proof of~\autoref{thm:main} suggests a simple way to reconstruct the original $MV$ assignment $\mu$, given the crease pattern $C$ and the forcing set $F$ produced by the {\sc ForcingSet} algorithm. First we construct a crimp forest from the sequence of distances induced by $C$ (note that the construction uses only interval lengths and is independent on the $MV$ assignment). By~\autoref{lem:forest}, the result is a crimp forest  $W_1$ isomorphic to the crimp forest $W$ used by the {\sc ForcingSet} to determine $F$. By~\autoref{lem:force}, crimpable sequences for  corresponding  nodes in $W_1$ and $W$ have the same number (and position) of creases that are in $F$. This implies that, for each 
crimpable sequence $\alpha$ of size $\ell$ corresponding to a node in $W_1$, at least $\lfloor \ell/2\rfloor$ creases of the same type (either $M$ or $V$) are in $F$ (as ensured by the {\sc ForcingSet} algorithm). Then we simply assign the remaining creases in $\alpha$ (the ones not in $F$) the opposite type.  Arguments similar to the ones used in the proof of~\autoref{thm:main} show that this MV assignment does not incur any conflicts for creases that occur in multiple crimpable sequences, and that it agrees with $\mu$.

\section{Conclusion}
This is the first paper that addresses the problem of finding a minimum forcing set in an arbitrary one-dimensional crease pattern. While folding problems can seem simple at first, the details of even one-dimensional folding are surprisingly complex. 
The algorithm presented in this paper finds a minimum forcing set for an arbitrary one-dimensional crease pattern in time linear in the 
number of creases.  
The algorithm complexities arise when there are consecutive sequences of equally
spaced creases, but such sequences should not be viewed as degenerate cases given that many useful origami crease patterns are tessellations. 
The most obvious direction for future work on forcing sets is to consider arbitrary two-dimensional crease patterns; some initial work has seen print on the Miura-ori crease pattern \cite{Ballinger:2015}.  Finding a general algorithm for the two-dimensional case is open.

\bibliographystyle{elsarticle-num} 
\bibliography{forcebib}

\section{Appendix}
This section is devoted to the proof of~\autoref{lem:order}, which is instrumental to the proof of correctness presented in~\autoref{sec:correct}, yet it is independent of the forcing set algorithm from~\autoref{sec:algorithm}  and can stand as a result on its own. We begin with a few insights into  folding maps, which tell us where points go when we fold sets of consecutive creases. 

\subsection{Folding Maps}
Recall that folding a crease $c_i$ means reflecting all points to one side of $c_i$  (i.e., all points to the right of $c_i$, or all points to the left of $c_i$, depending on $\mu(c_i)$) about $c_i$. 
\begin{figure}[hp]
\centerline{\includegraphics[width=0.75\linewidth]{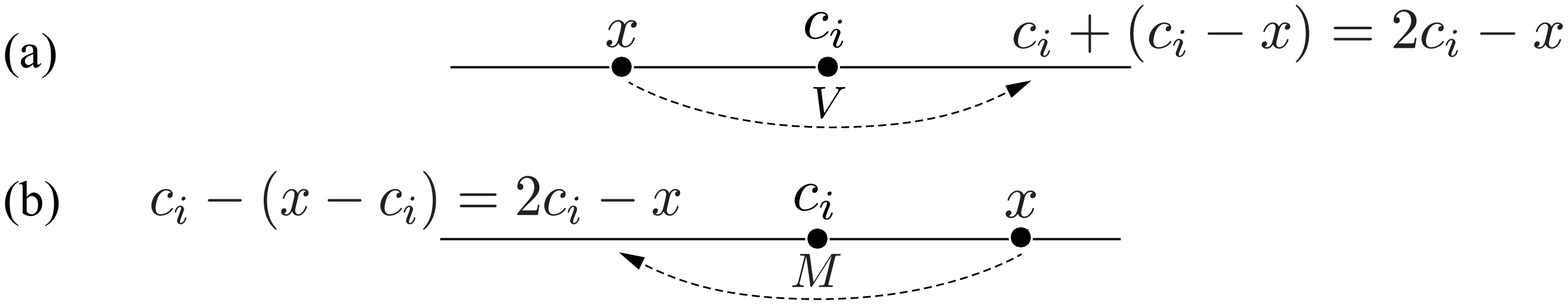}}
\caption{Reflecting point $x$ about crease $c_i$ (a) $x$ left of $c_i$ (b) $x$ right of $c_i$.}\label{fig:folding map}
\end{figure}
This operation (depicted in~\autoref{fig:folding map}) can be modeled by a function $\sigma_{c_i}:\mathbb{R}\rightarrow \mathbb{R}$ given by
\begin{equation}
\sigma_{c_i}(x) = c_i+(c_i-x) = 2c_i - x.
\label{eq:singlefold}
\end{equation}

\noindent
By composing such functions together, we can define a {\em folding map} that tells us where points go when we fold sets of consecutive creases.  If $S=(a_1,a_2,...,a_k)$ is a sequence of consecutive creases, then the {\em folding map $\sigma_{S,a_1}$ that fixes $[a_1,a_2]$} (i.e., it does not fold $a_1$) is 
\[
\sigma_{S,a_1}(x) = \sigma_{a_2} \sigma_{a_3} \sigma_{a_4} \cdots \sigma_{a_i}(x)\mbox{ where }x\in [a_i,a_{i+1}],~i < k.
\]
\begin{lemma}
Let $x$ be an arbitrary point on a line segment, and let $S = (a_1 , a_2 , ..., a_i )$ be a foldable crease pattern with the property that each fold about each crease $a_j$, with $2 \le j \le i$, affects the location of $x$. Then 
\begin{equation}
\sigma_{S,a_1}(x) =2a_2-2a_3+2a_4 -\cdots +(-1)^{i}2a_i + (-1)^{i+1}x.
\label{eq:foldingmap}
\end{equation}
\end{lemma}
\begin{proof}
The proof is by induction on the size $i$ of $S$ (which includes the fixed crease $a_1$). The base case correspond to $i = 2$ (that is, $S = (a_1, a_2)$ and the segment bends about crease $a_2$ only). In this case the location of $x$ after folding about $a_2$ is given by~(\ref{eq:singlefold}): $\sigma_{S,a_1}(x) = 2a_2-x = 2a_2+(-1)^3x$. So the base case holds.

The induction hypothesis is that, for any sequence $S$ of length $i$, the equality~(\ref{eq:foldingmap}) holds. For the inductive step, consider a sequence $S$ of length $i+1$. Let $S = (a_1, a_2, \ldots, a_{i}, a_{i+1})$ of , and let $T = (a_1, a_2, \ldots, a_{i})$ be the subsequence consisting of the first $i$ elements of $S$. 

Notice that the ordering in which the folds about the creases $a_j \in S$ are performed may affect the overlap 
order of the layers (intervals) in the final flat folding of $S$. For the purpose of our analysis, this nesting is not important; rather, it is the relative positioning of points in the folded state of $S$ (with respect to the fixed point $a_1$) that affects the location of $x$.
Arguments similar to the ones used in the proof of~\autoref{lem:crimp} show that any folding of the creases $a_2, a_3, \ldots, a_{i+1}$ places the points in the interval $[a_2, a_{i+1}]$ in the same positions, independent of the $MV$ assignment. The $MV$ assignment affects only the layering order and does not affect the positions of the creases in the folded state. 
This allows us to fold about the crease $a_{i+1}$ first, which by~(\ref{eq:singlefold}) changes the position of point $x$ to $x' = 2a_{i+1}-x$.  We can now view $x'$ as a new point, and apply the induction hypothesis to determine the new position of $x'$ after folding about each crease in $T$: 
\begin{eqnarray*}
   \sigma_{S, a_1}(x) &  = & \sigma_{T,a_1}(x') \\
                                 & = & 2a_2-2a_3+2a_4 -\cdots +(-1)^{i}2a_i + (-1)^{i+1}x' \\
                                 & = & 2a_2-2a_3+2a_4 -\cdots +(-1)^{i}2a_i + (-1)^{i+1}(2a_{i+1}-x) \\
                                 & = & 2a_2-2a_3+2a_4 -\cdots +(-1)^{i}2a_i + (-1)^{i+1}2a_{i+1} + (-1)^{i+2}x).                                 
\end{eqnarray*}
Thus the equality~(\ref{eq:foldingmap}) holds. 
\end{proof}

\noindent
Note that these folding maps do not depend on any $MV$ assignment.  They only tell us where points go after the creases are folded, regardless of the $MV$ assignment.

\subsection{Proof of~\autoref{lem:order}}

\orderlemma*
\begin{proof}
Our proof is by contradiction. Assume to the contrary that there is a crimp in $O_1$ that is not in $O_2$. Let $Z_1$ be the first crimp in the ordered sequence $O_1$ that does not occur in $O_2$. 
Let $C_1$ be the crease sequence produced by all crimps in $O_1$ prior to $Z_1$ (and so $Z_1$ is a crimp to be performed on $C_1$), and let $\alpha = (a_1, a_2, \ldots, a_k)$ be the crimpable sequence involved in the crimp $Z_1$. Let $E_2$ be the end crease sequence produced by $O_2$ (i.e., $E_2$ contains no crimpable sequences). Refer to~\autoref{fig:crimporder}.
We distinguish several cases, depending on whether one or more creases from the crease sequence $\alpha$  appear in $E_2$ or not. 
\begin{figure}[hptb]
\centerline{\includegraphics[width=.85\textwidth]{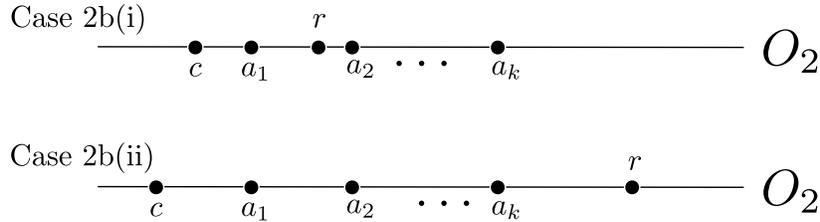}}
\caption{Case 1: all creases from $\alpha$ occur in $E_2$.}
\label{fig:crimporder}
\end{figure}

Before discussing these cases, we introduce some terminology and a new lemma that will be used in our case analysis. 
Define the \emph{interval distance} of the crimpable sequence 
$\alpha$ to be $|a_2-a_1|$ (which by definition equals $|a_{i+1}-a_i|$ for any $i=1,\ldots, k-1$).  
For ease of presentation, we also define the \emph{interval distance of a crimp} to be equal to the interval distance of the crease sequence involved in the crimp. 

\begin{lemma}\label{lem:intervaldistance}
Let $O_1$ and $O_2$ be two sequences of crimps performed on a crease pattern $C = (a_0, a_1, \ldots, a_n)$. Let $Q_1$ and $Q_2$ be the crease sequences produced by $O_1$ and $O_2$, respectively. If $a_i$ and $a_j$ are adjacent creases in both $Q_1$ and $Q_2$, for some $i, j \in [0, n]$, then the distance $|a_j - a_i|$ is the same in $Q_1$ and $Q_2$. 
\end{lemma}

\begin{proof}
Because $Q_1$ ($Q_2$) is the result of multiple crimps in $O_1$ ($O_2$), the creases $a_i$ and $a_j$ might not have started out in $C$ as being adjacent.  Let $d = |a_j - a_i|$ be the original distance between these two creases in $C$.  
If there are no creases between $a_i$ and $a_j$ in $C$ (i.e, $j = i+1$), then $a_i$ and $a_j$ have not been involved in any crimp in  $O_1$ ($O_2$), therefore the distance between $a_i$ and $a_j$ in $Q_1$ ($Q_2$) is also $d$. 


Suppose now that there is a sequence of creases $b_1, ..., b_k$ between $a_i$ and $a_j$ in $C$. Because $a_i$ and $a_{j}$ are adjacent in $Q_1$ ($Q_2$), the creases $b_j$, $j = 1, \ldots, k$ must have been removed by crimps in $O_1$ ($O_2$). 
%
%
Let $B = (a_i, b_1, ..., b_k)$. The folding map~(\ref{eq:foldingmap}) shows us that no matter what crimps in $O_1$ or $O_2$ removed the creases $b_j$, the resulting distance between $a_i$ and $a_{i+1}$ will be 
\[
|\sigma_{B,a_i}(a_{j}) - a_i| = |2b_1 - 2b_2 + \cdots +(-1)^k 2b_k +(-1)^{k+1}a_{j} - a_i|.
\]
Therefore the distance between $a_i$ and $a_j$ is the same in $Q_1$ and $Q_2$. 
\end{proof}

\noindent
We now turn to our case analysis on the intersection between the crease sequence $\alpha$ and the end sequence $E_2$. 

\vspace{.1in}
\noindent{\bf Case 1:}  All creases from $\alpha$ occur in $E_2$: $a_1, ..., a_k \in E_2$.
Assume first that $a_1, ..., a_k$ are consecutive in $E_2$.  (See~\autoref{fig:crimporder}.) By~\autoref{lem:intervaldistance}, the distance between $a_i$ and $a_{i+1}$ in $E_2$ is equal to the distance between $a_i$ and $a_{i+1}$ in $C_1$, for $i=1, ..., k-1$.  (Recall that $C_1$ is the crease sequence produced by all crimps in $O_1$ prior to $Z_1$.) 
Because $E_2$ is an end sequence, we can assume without loss of generality that distances between adjacent creases from the left end of the paper up to $a_1$ are monotonically increasing (otherwise, we can reflect the paper about the left end and reverse the labeling of the $\alpha$ sequence, so that this assumption holds). 

In the crease sequence $E_2$, let $x_2$ be the distance from $a_1$ to the crease $b_2$ on its left (note that $b_2$ may be the left end of the paper, if no other creases exist to the left of $a_1$).  By our assumption, $x_2 \le a_2-a_1$.  In $C_1$ however, the distance $x_1$ from $a_1$ to the crease $b_1$ on its left (which again could be the left end of the paper) is strictly greater than $a_2-a_1$.  (This is a necessary condition for the crimp $Z_1$ to be possible, or simply by definition of a crimp.)  So the inequalities 
$x_1 > a_2-a_1 \ge x_2$ hold, therefore $x_1 > x_2$. By~\autoref{lem:crimpgrow}, the interval to the left of $a_1$ in the initial crease sequence $C$ is no longer than $x_2$ (because it may only increase or stay the same in size with each crimp). This implies that some crimp $Z$ must have occurred in $O_1$ prior to $Z_1$ that lengthened the interval to the left of $a_1$ from some size below $x_1$ to size $x_1$. Note that the crimp $Z$ could not occur in $O_2$, because the interval to the left of $a_1$ in $E_2$ is shorter than the one produced by the crimp $Z$. This contradicts our choice of $Z_1$ (because $Z$ comes before $Z_1$ in the sequence $O_1$, and $Z \notin O_2$). 

Assume now that there exists an $i$ such that $a_i$ and $a_{i+1}$ are not consecutive in $E_2$.  Let $c$ be a crease in $E_2$ between $a_i$ and $a_{i+1}$.  Since $c \notin \alpha$, some crimp $Z'$ must have occurred in $O_1$ prior to $Z_1$ to remove $c$.  Now note that $Z' \notin O_2$, because $c \in E_2$ and $Z'$ removes $c$.  This again contradicts our choice of $Z_1$, because $Z' \notin O_2$ occurs before $Z_1$. 

\vspace{.1in}

\noindent{\bf Case 2:}  There exists an $i$ such that $a_i\not\in E_2$.
Let $i$ be such that $a_i$ is the first crease from among all creases in the sequence $\alpha$ to be removed by a crimp $Z_2 \in O_2$ (thus $a_i \not\in E_2$).  Consider the crease sequence $C_2$ produced by all crimps in $O_2$ prior to $Z_2$ (and so $Z_2$ is a crimp to be performed on $C_2$). Note that all creases $a_1, a_2, \ldots, a_k$ are in $C_2$, because $a_i$ is the first of these creases to be removed, and $C_2$ is the sequence prior to removing $a_i$. Because at least two adjacent creases are removed by a crimp, there must have been a neighbor $c$ of $a_i$ that was removed by $Z_2$ along with $a_i$ (by the definition of a crimp). 


\begin{figure}[hptb]
\centerline{\includegraphics[width=.9\textwidth]{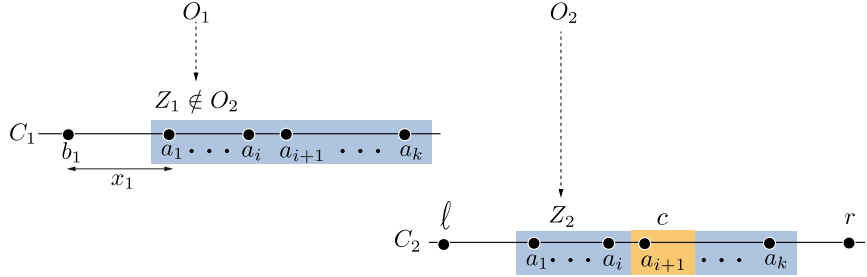}}
\caption{Case 2a: $a_i$ and $c$ removed by the crimp $Z_2$.}
\label{fig:crimporder2}
\end{figure}

\noindent
{\bf Case 2a:}  $c=a_{i-1}$ or $c=a_{i+1}$. (See~\autoref{fig:crimporder2}).
Assume without loss of generality that $c=a_{i+1}$.  By~\autoref{lem:intervaldistance}, the interval distance $d=a_{i+1}-a_{i}$ of $Z_1$ equals the interval distance $a_{i+1}-a_i$ of $Z_2$.  By the definition of a crimp, the intervals immediately to the left/right of the creases involved in $Z_2$ must both be strictly longer than $a_{i+1}-a_i$.  

Let $r$ be the crease in $C_2$ immediately to the right of the creases involved in $Z_2$.  We argue that $r$ lies to the right of $a_k$.  Otherwise, if $\exists j>i$ such that $r$ lies between $a_j$ and $a_{j+1}$, then 
\[
r-a_j \leq d=a_{i+1}-a_i, 
\]
because by~\autoref{lem:crimpgrow} the interval to the right of $a_j$ can only increase or stay the same in size until it reaches the limit $a_{j+1}-a_j = d$ in $C_1$.  This contradicts the requirement that $r-a_j> d$ for the crimp $Z_2$ to take place.  Hence $r$ lies to the right of $a_k$, and similarly the crease $\ell$ immediately to the left of the creases involved in $Z_2$ lies to the left of $a_1$.  Similar arguments show that there are no creases in $C_2$ between $a_j$ and $a_{j+1}$, for each $j = 1, 2, \ldots, k-1$.

Next we show that there are no creases in $C_2$ between $a_k$ and $r$, or between between $\ell$ and $a_1$.  This would imply that $Z_2$ is equivalent to $Z_1$, contradicting our assumption that $Z_1 \notin O_2$.  

Assume to the contrary that there exists a crease $w$ between $a_k$ and $r$ in $C_2$.  Because $r$ lies immediately to the right of the creases involved in $Z_2$, we have that $w-a_k = d$ and $r-w>d$ (because $Z_2$ is a crimp). Note however that $w$ is not part of the sequence $\alpha$ in $C_1$, therefore $w$ must have been removed by a crimp $Z_w \in O_1$ prior to $Z_1$. The crimp $Z_w$ cannot occur in $O_2$, because $Z_w$ removes $w$ but not $a_i$, and in $O_2$ the crease $w$ gets removed along with $a_i$ by the crimp $Z_2$. 
This contradicts our choice of $Z_1$ (because $Z_w \notin O_2$ occurs prior to $Z_1$). Similar arguments show that there are no creases between $\ell$ and $a_1$ in $C_2$. 

So the crimp $Z_2 \in O_2$ involves only $a_1, a_2, ..., a_k$ and we have that $Z_1$ is the same as $Z_2$.  This contradicts our assumption that $Z_1 \notin O_2$. 

\vspace{.1in}

\begin{figure}[hptb]
\centerline{\includegraphics[width=.75\textwidth]{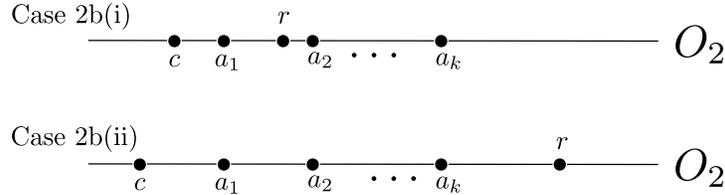}}
\caption{Case 2b, depending on the location of crease $r$ immediately to the right of the creases involved in $Z_2 \in O_2$.}\label{fig4}
\end{figure}

{\bf Case 2b:}  $c\not= a_{i-1}$ and $c\not= a_{i+1}$.
Assume first that $i>1$ and $c\in (a_{i-1}, a_i)$, or $i < k$ and $c\in (a_{i}, a_{i+1})$.  Then in $O_1$ some crimp $Z_c$ prior to $Z_1$ must have removed $c$.  Note that $Z_c \notin O_2$, because $Z_c$ removes $c$ but does not remove $a_i$, and in $O_2$ the crease $c$ gets removed along with $a_i$ by $Z_2$. This contradicts our choice of $Z_1$ as being the first crimp in $O_1$ that does not occur in $O_2$.

The only case left is $a_i=a_1$ and $c$ lies left of $a_1$ (or the symmetric case $a_i=a_k$ and $c$ lies right of $a_k$, which is similar).  
%
%
As in Case 2a, let $r$ be the crease immediately to the right of the creases involved in $Z_2$.  Then $r$ may lie either in the interval $(a_1, a_2]$, or strictly to the right of $a_2$. 

Assume first that $r \in (a_1, a_2]$ (see~\autoref{fig4} top). Then $r-a_1\leq d$, because by~\autoref{lem:crimpgrow} the interval immediately to the right of $a_1$ can only increase or stay the same in size with each crimp, and the limit $d = |a_2-a_1|$ is reached when $a_2$ becomes  adjacent to $a_1$ (without being separated by $r$).  This along with the fact that $Z_2$ is a crimp implies that 
$a_1-c < r-a_1\leq d$.  
This means that $a_1-c < d < a_1-b$, where $b$ is the crease immediately to the left of $a_1$ in $C_1$.  Then $c$ must have been removed in $O_1$ by a crimp $Z'_c$ prior to $Z_1$. Now note that $Z'_c$ eliminates $c$ but does not eliminate $a_1$, and $Z_2$ eliminates both $c$ and $a_1$. This implies that $Z'_c \notin O_2$. contradicting our choice of $Z_1$. 

Assume now that $r$ is strictly to the right of $a_2$ (see~\autoref{fig4} bottom). Arguments identical to those used in case 2a show that $r$ lies strictly to the right of $a_k$, and no creases exist in $C_2$ between $a_j$ and $a_{j+1}$, or between $a_k$ and $r$.  It follows  that $a_1, a_2, ..., a_k$ are all involved in the crimp $Z_2 \in O_2$, which is identical to $Z_1 \in O_1$. This contradicts our assumption that $Z_1 \notin O_2$. 
\end{proof}

\end{document}